\documentclass[12pt]{article}
\usepackage{amsmath,amssymb,amsthm}
\usepackage{graphicx,psfrag,epsf}
\usepackage{enumerate}
\usepackage[numbers]{natbib}
\usepackage{url} 
\usepackage{lipsum}
\usepackage{braket,amsfonts,amsopn} 

\usepackage[capitalise]{cleveref}
\crefformat{figure}{#2Figure~#1#3}
\DeclareMathOperator{\E}{E} 
\DeclareMathOperator{\Cov}{Cov} 
\DeclareMathOperator{\argmax}{argmax} 
\DeclareMathOperator{\diag}{diag} 
\DeclareMathOperator{\Ra}{Ra} 
\DeclareMathOperator{\Ker}{Ker} 
\DeclareMathOperator{\spans}{span} 
\DeclareMathOperator{\dist}{dist} 
\DeclareMathOperator{\Nprmse}{Normalized~PRMSE}

\DeclareMathOperator{\N}{N}

\usepackage{multirow,color}

\newtheorem{proposition}{Proposition}

\begin{document}

\def\spacingset#1{\renewcommand{\baselinestretch}%
{#1}\small\normalsize} \spacingset{1}

\title{\bf Dimension Reduction for Gaussian Process Emulation: An Application to the Influence of Bathymetry on Tsunami Heights}
\author{Xiaoyu Liu and Serge Guillas\\
  Department of Statistical Science, University College London}
\maketitle

\bigskip
\begin{abstract}
High accuracy complex computer models, also called simulators, require large resources in time and memory to produce realistic results. Statistical emulators are computationally cheap approximations of such simulators. They can be built to replace simulators for various purposes, such as the propagation of uncertainties from inputs to outputs or the calibration of some internal parameters against observations. However, when the input space is of high dimension, the construction of an emulator can become prohibitively expensive. In this paper, we introduce a joint framework merging emulation with dimension reduction in order to overcome this hurdle. The gradient-based kernel dimension reduction technique is chosen due to its ability to drastically decrease dimensionality with little loss in information. The Gaussian process emulation technique is combined with this dimension reduction approach. Theoretical properties of the approximation are explored. Our proposed approach provides an answer to the dimension reduction issue in emulation for a wide range of simulation problems that cannot be tackled using existing methods. The efficiency and accuracy of the proposed framework is demonstrated theoretically, and compared with other methods on an elliptic partial differential equation (PDE) problem. We finally present a realistic application to tsunami modeling. The uncertainties in the bathymetry (seafloor elevation) are modeled as high-dimensional realizations of a spatial process using a geostatistical approach. Our dimension-reduced emulation enables us to compute the impact of these uncertainties on resulting possible tsunami wave heights near-shore and on-shore. Considering an uncertain earthquake source, we observe a significant increase in the spread of uncertainties in the tsunami heights due to the contribution of the bathymetry uncertainties to the overall uncertainty budget. These results highlight the need to include the effect of uncertainties in the bathymetry in tsunami early warnings and risk assessments.
\end{abstract}

\noindent%
\vfill

\section{Introduction}
\label{sec:intro}
Simulators are widely employed to reproduce physical processes and explore their behavior, in fields such as fluid dynamics or climate modeling. To characterize the impact of the uncertainties in the boundary conditions or the parameterizations 
of the underlying physical processes, a sufficient number of simulations are required. However, when the simulators are computationally expensive, as it is the case for high accuracy simulations, the task can become extremely costly or even prohibitive. One prevailing way to overcome this hurdle is to construct statistical surrogates, namely emulators, to approximate the computer simulators in a probabilistic way \cite{sacks1989design}. Emulators are trained on a relatively small number of well chosen simulations, i.e. a design of computer experiments.  Outputs at any input can be predicted at little computational cost with emulators. One can then employ emulators for any subsequent purposes such as uncertainty propagation, sensitivity analysis and calibration.

With high-dimensional inputs, say beyond 20 dimensions (usually in the hundreds or thousands), a large design is usually required to explore the input space, typically in the order of 10 times the number of dimensions, for a reasonable level of approximation. One would face serious computational problems since the original simulator cannot be run many times or is very expensive to run. Advanced designs such as Latin Hypercubes or new sequential designs \cite{beck2016sequential} that are more efficient than Latin Hypercubes only partially alleviate the issue. As a result, methods that adequately reduce the dimension of the input space are required, as high-dimensional inputs are often present in computer models, e.g. as boundary conditions like the bathymetry (i.e. seafloor elevation) in tsunami modeling. Some approaches ignore high-dimensional inputs and add stochastic terms to account for their contribution \cite{iooss2009global,marrel2012global}. These methods are easy to implement and effective in some applications. However, repeated simulations at the 
same input parameters that are encoded in the emulation are often required to estimate the variability due to those parameters that are ignored. The variability estimates are often restricted to the second moments, and the input-output relationships over the ignored inputs are not clear. Constantine et al. \cite{constantine2014active} proposed to find rotations of the input space with the strongest variability in the gradients of the simulators and constructed a response surface on such a low-dimensional active subspace. This Active Subspace (AS) approach has been demonstrated to be effective theoretically and numerically. Constantine and Gleich \cite{constantine2015computing} studied further the properties of the Monte Carlo approximation of the subspace. However, this method requires the calculation of a sufficient number of gradients explicitly, which unfortunately prevents its use in many applications. The gradients are often unavailable in many realistic simulators, and typically intractable for systems of mixed PDEs or multi-physics simulations. Even in the rare situations where gradients are computable numerically, the computational cost of obtaining them could be prohibitive.

The concept of active subspace is closely related to the sufficient dimension reduction (SDR) \cite{cook1994interpretation,cook2009regression} and effective dimension reduction (EDR) \cite{li1991sliced} in the statistical community. Given an explanatory variable $\mathbf{X}\in\mathbb{R}^m$ (input) and response variable $Y$ (output), the aim of SDR (or EDR) is to find the directions in the subspace of $\mathbf{X}$ that contain sufficient information about the response for statistical inference. More specifically, a SDR $R(\mathbf{X})\in\mathbb{R}^d$ where $d<m$ satisfies $p(Y|\mathbf{X})=\tilde{p}(Y|R(\mathbf{X})),$ where $p(Y|\mathbf{X})$ and $\tilde{p}(Y|R(\mathbf{X}))$ are conditional probability density functions with respect to $\mathbf{X}$ and $R(\mathbf{X})$ respectively. The EDR approach aims to specifically find a linear projection matrix $\mathbf{B}$ onto a $d$-dimensional subspace $(d<m)$ such that $\mathbf{B}^{T}\mathbf{B}=\mathbf{I}_d$ and 
\begin{equation}\label{prob-eq}
p(Y|\mathbf{X})=\tilde{p}(Y|\mathbf{B}^{T}\mathbf{X})~~~\textrm{or equivalently}~~~ Y\perp X|\mathbf{B}^{T}X.
\end{equation}
Several methods have been developed to find SDR including nonparametric approaches such as sliced inverse regression (SIR) \cite{li1991sliced}, minimum average variance estimation (MAVE) \cite{xia2002adaptive}, and parametric approaches like 
principal fitted components (PFC) \cite{cook2007fisher,cook2008principal}. In this paper, we adopt the gradient-based kernel dimension reduction (gKDR) \cite{fukumizu2014gradient} to construct low-dimensional approximations to the simulators. The gKDR approach does not require any strong assumptions on the variables and distributions. The response variable can be of arbitrary type: continuous or discrete, univariate or multivariate. Unlike the active subspace method, gradients are not required to be computed explicitly but are estimated non-parametrically and implicitly using stable kernel methods. Our proposed approach therefore provides an answer to the dimension reduction issue in emulation for a wide range of problems that cannot be tackled using existing methods at the moment. Moreover, the gKDR approach ends up with an eigen-problem without any needs of elaborate numerical optimization and thus can be applied to large and high-dimensional problems.

We introduce a joint framework to approximate the high-dimensional simulators by combining statistical emulators with the gKDR approach. Deterministic simulators are considered here, however the framework could potentially be applied to stochastic simulators, with additional treatments to the stochastic effect in the emulation, see e.g. \cite{henderson2009bayesian}. Throughout this paper, the mainstream Gaussian process (GP) emulators are employed for illustration. But the general framework and most of the results would potentially hold for other emulation techniques.

The paper is organized as follows. \Cref{sec:gp} and \ref{sec:gkdr} review GP emulation and the gKDR approach respectively. In \cref{sec:model}, a joint framework of dimension reduction combined with emulation is proposed and some theoretical properties are established. \Cref{sec:sim} contains the numerical experiments on an elliptic PDE and an application to the propagation of uncertainties in the bathymetry to tsunami wave heights, as well as comparison with alternative methods. \Cref{sec:dis} consists of some conclusive discussion.

\section{Gaussian process emulation}
\label{sec:gp}
A Gaussian process is a collection of random variables such that any finite subset of these variables follow a joint Gaussian distribution \cite{rasmussen2006gaussian}. It is widely used in various scientific fields. Here we briefly review some basics of its application in statistical emulation.

A deterministic simulator with multivariate input $\mathbf{X}=(x_1,...,x_m)^T\in\mathbb{R}^m$ and univariate output $Y\in\mathbb{R}$ can be represented as $Y=f(\mathbf{X})$. The GP emulator assumes that the simulator output $Y=f(\mathbf{X})$ can be modeled with a Gaussian process. It is commonly assumed that the mean can be written as $\E(Y)=m(\mathbf{X})=\mathbf{h}^T(\mathbf{X})\boldsymbol{\beta}$, where $\mathbf{h}(\mathbf{X})$ is a $q$-vector of pre-defined regression functions and the coefficients $\boldsymbol{\beta}\in\mathbb{R}^q$. In practice, a constant or linear form for the regression functions would perform well. The covariance between two simulator outputs $Y=f(\mathbf{X})$ and $Y'=f(\mathbf{X}')$ is usually represented as $\Cov(Y,Y')=k(\mathbf{X},\mathbf{X}')=\sigma^2 c(\mathbf{X},\mathbf{X}')$, 
where the positive scalar parameter $\sigma^2$ is the process variance and $c(\mathbf{X},\mathbf{X}')$ is the correlation function. One common choice for the correlation function is the squared-exponential correlation $c(\mathbf{X},\mathbf{X}')=\prod_{i=1}^m \exp\left(-(x_i-x_i')^2/\delta_i^2\right),$
where $\boldsymbol{\delta}=(\delta_1,...,\delta_m)^T\in (0,\infty)^m$ controls the correlation lengths. 

Suppose that the simulator is run at $n$ inputs $\mathbf{X}_1,...,\mathbf{X}_n$ and the outputs are $Y_1,...,Y_n$ respectively. We may need to impose a prior for the parameter $\boldsymbol{\beta}$ in the mean function  $m(\mathbf{X})=\mathbf{h}^T(\mathbf{X})\boldsymbol{\beta}$. One of the popular choices is a Gaussian prior, $\boldsymbol{\beta}\sim \N(\mathbf{b},\mathbf{V})$, which forms a conjugate prior with the GP likelihood. At any $n^*$ desired inputs $\mathbf{X}_1^*,\mathbf{X}_2^*,...,\mathbf{X}_{n^*}^*$, the respective outputs are denoted by $Y_1^*,Y_2^*,...,Y_{n^*}^*$. Then letting $\mathbf{H}=(\mathbf{h}(\mathbf{X}_1),...,\mathbf{h}(\mathbf{X}_n))$ and $\mathbf{H}^{*}=(\mathbf{h}(\mathbf{X}_1^*),...,\mathbf{h}(\mathbf{X}_n^*))$, the predictive process of $\mathbf{Y}^{*}$ (also known as kriging prediction) given the observed data and the covariance function is
\begin{equation}
\mathbf{Y}^*|\mathbf{Y},k(\cdot,\cdot;\boldsymbol{\theta})\sim \N(\hat{\mathbf{m}}^*,\hat{\boldsymbol{\Sigma}}^{*}),
\end{equation}
with
\begin{equation}
\hat{\mathbf{m}}^*=\mathbf{H}^{*T}\hat{\boldsymbol{\beta}}+\mathbf{K}_{*}\mathbf{K}^{-1}(\mathbf{Y}-\mathbf{H}^{T}\hat{\boldsymbol{\beta}}),
\end{equation}
and
\begin{equation}
\hat{\boldsymbol{\Sigma}}^*=\mathbf{K}_{**}+\mathbf{P}^{T}(\mathbf{V}^{-1}+\mathbf{H}\mathbf{K}^{-1}\mathbf{H}^{T})^{-1}\mathbf{P},
\end{equation}
where $\boldsymbol{\theta}$ represents the hyperparameters from the covariance function, $\mathbf{K}$, $\mathbf{K}_{*}$ and $\mathbf{K}_{**}$ are $n\times n$, $n^{*}\times n$ and $n^{*}\times n^{*}$ matrices respectively with the associated $(i,j)$-th entry as $\mathbf{K}(i,j)=k(\mathbf{X}_i,\mathbf{X}_j)$, $\mathbf{K}_{*}(i,j)=k(\mathbf{X}_i^*,\mathbf{X}_j)$ and $\mathbf{K}_{**}(i,j)=k(\mathbf{X}_i^*,\mathbf{X}_j^*)$, $\mathbf{P}=\mathbf{H}^*-\mathbf{H}\mathbf{K}^{-1}\mathbf{K}_{*}^{T}$ and $\hat{\boldsymbol{\beta}}=(\mathbf{V}^{-1}+\mathbf{H}\mathbf{K}^{-1}\mathbf{H}^{T})^{-1}(\mathbf{H}\mathbf{K}^{-1}\mathbf{Y}+\mathbf{V}^{-1}\mathbf{b})$.

We can see that the output at any desired input predicted using a GP emulator is a distribution rather than a single value. This could be used to estimate the uncertainty introduced into the prediction with emulator and to evaluate the confidence about the prediction. However, the hyperparameters $\boldsymbol{\theta}$ is usually unknown and needs to be specified properly. It is possible to make a fully Bayesian inference with appropriate prior $\pi(\boldsymbol{\theta})$. But this usually requires costly MCMC approach for the analytically intractable posterior. In practice, a computationally cheap alternative is often employed by specifying the hyperparameters $\boldsymbol{\theta}$ at the most probable values. This could be done by maximizing the marginal likelihood,
\begin{align}
\begin{split}
L(\boldsymbol{\theta})&=\log p(\mathbf{Y}|\mathbf{b},\mathbf{V})\\&
=-\frac{1}{2}(\mathbf{H}^{T}\mathbf{b}-\mathbf{Y})^{T}(\mathbf{K}+\mathbf{H}^{T}\mathbf{V}\mathbf{H})^{-1}(\mathbf{H}^{T}\mathbf{b}-\mathbf{Y})
-\frac{1}{2}\log |\mathbf{K}+\mathbf{H}^{T}\mathbf{V}\mathbf{H}|-\frac{n}{2}\log 2\pi.
\end{split}
\end{align}
Then the prediction can be completed by plugging in $\hat{\boldsymbol{\theta}}=\argmax_{\boldsymbol{\theta}}L(\boldsymbol{\theta})$.

Usually there is no sufficient information about the parameter $\boldsymbol{\beta}$, hence a vague prior can be imposed by letting $\mathbf{V}^{-1}\rightarrow \mathbf{O}$ and $\mathbf{b}=\mathbf{0}$, where $\mathbf{O}$ is the matrix of zeros. In this case, the conditional mean and covariance of the predictive process are respectively 
\begin{equation}
\hat{\mathbf{m}}^*=\mathbf{H}^{*T}\hat{\boldsymbol{\beta}}+\mathbf{K}_{*}\mathbf{K}^{-1}(\mathbf{Y}-\mathbf{H}\hat{\boldsymbol{\beta}}),
\end{equation}
and
\begin{equation}
\hat{\boldsymbol{\Sigma}}^*=\mathbf{K}_{**}+\mathbf{P}^{T}(\mathbf{H}\mathbf{K}^{-1}\mathbf{H}^{T})^{-1}\mathbf{P},
\end{equation}
where $\hat{\boldsymbol{\beta}}=(\mathbf{H}\mathbf{K}^{-1}\mathbf{H}^{T})^{-1}\mathbf{H}\mathbf{K}^{-1}\mathbf{Y}$.
This is closely related to the $t$-process \cite{o1994kendall} when a weak prior for $(\boldsymbol{\beta},\sigma^2,\boldsymbol{\delta})$ that $\pi(\boldsymbol{\beta},\sigma^2,\boldsymbol{\delta})\propto \sigma^{-2}\pi_{\boldsymbol{\delta}}(\boldsymbol{\delta})$ is assumed with the mean function $m(\cdot)=\mathbf{h}^{T}(\cdot)\boldsymbol{\beta}$ and the covariance function $k(\cdot,\cdot)=\sigma^2 c(\cdot,\cdot;\boldsymbol{\delta})$, where $\boldsymbol{\delta}$ contains the parameters in the correlation function $c(\cdot,\cdot)$.

When $k(\mathbf{X},\mathbf{X}')=\sigma^2 c(\mathbf{X},\mathbf{X}')$ is used with a continuous correlation function, such as the squared-exponential correlation, the emulator interpolates through the training data, i.e. $\hat{m}(\mathbf{X}_i)=Y_i$ and $\hat{v}(\mathbf{X}_i)=0$ at the training points $\{\mathbf{X}_i\}_{i=1}^n$. When a nugget term is included, this is no longer true. A nugget term can be included, e.g. to mitigate numerical instabilities or account for the stochastic terms in simulations \cite{andrianakis2012effect}. The correlation function $c(\mathbf{X},\mathbf{X}')$ can be extended with the addition of a nugget as  $\tilde{c}(\mathbf{X},\mathbf{X}')=\nu I_{\mathbf{X}=\mathbf{X}'}+(1-\nu)c(\mathbf{X},\mathbf{X}'),$
where $\nu>0$ is the nugget term, and $I_{\mathbf{X}=\mathbf{X}'}$ is the indicator function that takes $1$ if $\mathbf{X}=\mathbf{X}'$ and $0$ otherwise. The associated correlation matrix is $\tilde{\mathbf{K}}=(1-\nu)\mathbf{K}+\nu \mathbf{I}$, where $\mathbf{I}$ is the identity matrix.

In practice, the error in the prediction of GP emulator depends on the number of training data points. As there are more and more training data points, the GP emulator will be expected to recover the simulator. There are several theoretical results on how well the GP emulator $\hat{f}$ can approximate the simulator $f$ in the literature. For example, given $n$ training samples that are quasi-uniformly distributed on $\Omega\subset\mathbb{R}^{d}$, the error can be bounded \cite{fasshauer2011positive} as $\|f-\hat{f}\|_{\infty}\leq C_d n^{-p/d}\|f\|_{\mathcal{H}}$ for any $f$ in some function space $\mathcal{H}$ over $\Omega$, typically a Hilbert space, where $p$ controls the smoothness of the function and $C_d$ is a constant depending on the dimension $d$. This result suggests that $\hat{f}$ provides arbitrarily high approximation order when $p=\infty$, i.e. $f$ is infinitely smooth. However, this rate decreases as the dimension increases and the constant $C_d$ also grows with $d$. Hence, the approximation deteriorates in very high dimensions. This implies that more evaluations of the simulator are required to train an accurate emulator when the number of input parameters $d$ increases and the associated computational cost of constructing an emulator could increase dramatically as a result. Therefore, it is desirable to reduce the dimension of the problem from the perspectives of both accuracy and efficiency.

\section{Gradient-based kernel dimension reduction}
\label{sec:gkdr}
For a set $\Omega$, a symmetric kernel $k:\Omega\times\Omega\rightarrow\mathbb{R}$ is positive-definite if $\sum_{i,j=1}^n c_ic_jk(\omega_i,\omega_j)\geq 0$ for any $\omega_1,...,\omega_n\in\Omega$ and $c_1,...,c_n\in\mathbb{R}$. Then a positive-definite kernel $k$ on $\Omega$ is uniquely associated with a Hilbert space $\mathcal{H}$ consisting of functions on $\Omega$ such that 1) $k(\cdot,\omega)\in\mathcal{H}$; 2) the linear hull of $\{k(\cdot,\omega)|\omega\in\Omega\}$ is dense in $\mathcal{H}$; 3) $\langle h,k(\cdot,\omega)\rangle_{\mathcal{H}}=h(\omega)$ for any $\omega\in\Omega$ and $h\in\mathcal{H}$ where $\langle\cdot,\cdot\rangle_{\mathcal{H}}$ is the inner product in $\mathcal{H}$. Because the third property implies that the kernel $k$ reproduces any function $h\in\mathcal{H}$, the Hilbert space $\mathcal{H}$ is called the reproducing kernel Hilbert space (RKHS) associated with $k$. Let $(\mathbf{X},Y)$ be a random vector on the domain $\mathbb{R}^m\times\mathcal{Y}$, and $k_{\mathcal{X}}$ and $k_{\mathcal{Y}}$ be positive definite-kernels on $\mathbb{R}^m$ and $\mathcal{Y}$ with respective RKHS $\mathcal{H}_{\mathcal{X}}$ and $\mathcal{H}_{\mathcal{Y}}$. We shortly present the salient facts about the gKDR approach.

Fukumizu and Leng \cite{fukumizu2014gradient} noted that for any $g\in \mathcal{H}_{\mathcal{Y}}$, there exists a function $\varphi_g(\mathbf{z})$ on $\mathbb{R}^d$ such that 
\begin{equation}
\E\left[g(Y)|\mathbf{X}\right]=\varphi_g(\mathbf{B}^{T}\mathbf{X}).
\end{equation}
Then, under mild assumptions, we have, for any $\mathbf{X}=\mathbf{x}$,
\begin{equation}
\frac{\partial}{\partial x_{i}}\E\left[g(Y)|\mathbf{X}=\mathbf{x}\right]=\sum\limits_{a=1}^d \mathbf{B}_{ia}\langle g,\nabla_a\varphi(\mathbf{B}^{T}\mathbf{x})\rangle_{\mathcal{H}_{\mathcal{Y}}}.
\end{equation}

On the other hand, defining the cross-covariance operator $C_{Y\mathbf{X}}:\mathcal{H}_{\mathcal{X}}\rightarrow \mathcal{H}_{\mathcal{Y}}$ as the operator such that
\begin{equation}
\langle h_2,C_{Y\mathbf{X}}h_1\rangle_{\mathcal{H}_{\mathcal{Y}}}=\E\left[h_1(\mathbf{X})h_2(Y)\right],
\end{equation}
holds for all $h_1\in\mathcal{H}_{\mathcal{X}}$, $h_2\in\mathcal{H}_{\mathcal{Y}}$, and using the fact that
\begin{equation}
C_{\mathbf{X}\mathbf{X}}\E[g(Y)|\mathbf{X}]=C_{\mathbf{X}Y}g,
\end{equation}
if $\E[g(Y)|\mathbf{X}]\in\mathcal{H}_{\mathcal{X}}$ for any $g\in\mathcal{H}_{\mathcal{Y}}$ \cite{fukumizu2004dimensionality}, we obtain
\begin{equation}
\frac{\partial}{\partial x_{i}}\E\left[g(Y)|\mathbf{X}=\mathbf{x}\right]=\left\langle g,C_{Y\mathbf{X}}C_{\mathbf{X}\mathbf{X}}^{-1}\frac{\partial k_{\mathcal{X}(\cdot,\mathbf{x})}}{\partial x_{i}}\right\rangle_{\mathcal{H}_{\mathcal{Y}}}.
\end{equation}

Equating the two expressions above yields for $i,j=1,...,m$,
\begin{align}
\begin{split}
\mathbf{M}_{ij}(\mathbf{x})&=\left\langle C_{Y\mathbf{X}}C_{\mathbf{X}\mathbf{X}}^{-1}\frac{\partial k_{\mathcal{X}(\cdot,\mathbf{x})}}{\partial x_{i}}, C_{Y\mathbf{X}}C_{\mathbf{X}\mathbf{X}}^{-1}\frac{\partial k_{\mathcal{X}(\cdot,\mathbf{x})}}{\partial x_{j}}\right\rangle_{\mathcal{H}_{\mathcal{Y}}}\\
&=\sum\limits_{a,b=1}^d \mathbf{B}_{ia}\mathbf{B}_{jb}\langle \nabla_a\varphi(\mathbf{B}^{T}\mathbf{x}),
\nabla_b\varphi(\mathbf{B}^{T}\mathbf{x})\rangle_{\mathcal{H}_{\mathcal{Y}}}.
\end{split}
\end{align}
Therefore, the dimension reduction projection matrix $\mathbf{B}$ is formed as the eigenvectors associated with the nontrivial eigenvalues of the $m\times m$ matrix $\mathbf{M}(\mathbf{x})$.

Given i.i.d. samples $(\mathbf{X}_1,Y_1), ..., (\mathbf{X}_n,Y_n)$, the matrix $\mathbf{B}$ can be approximated with $\tilde{\mathbf{B}}$ \cite{fukumizu2014gradient} that contains the first $d$ eigenvectors of the following $m\times m$ symmetric matrix,
\begin{equation}
\tilde{\mathbf{M}}_n=\frac{1}{n}\sum\limits_{i=1}^n \nabla\mathbf{k}_\mathbf{X}(\mathbf{X}_i)^{T}(\mathbf{G}_\mathbf{X}+n\epsilon_n\mathbf{I})^{-1}\mathbf{G}_Y(\mathbf{G}_\mathbf{X}+n\epsilon_n\mathbf{I})^{-1}\nabla\mathbf{k}_\mathbf{X}(\mathbf{X}_i),
\end{equation} 
where $\mathbf{G}_\mathbf{X}$ and $\mathbf{G}_Y$ are the Gram matrices with the $(i,j)$-entry as $k_{\mathcal{X}}(\mathbf{X}_i,\mathbf{X}_j)$ and $k_{\mathcal{Y}}(Y_i,Y_j)$ respectively, $\nabla\mathbf{k}_\mathbf{X}(\cdot) = (\partial k_{\mathcal{X}}(\mathbf{X}_1,\cdot)/\partial \mathbf{x},...,\partial k_{\mathcal{X}}(\mathbf{X}_n,\cdot)/\partial \mathbf{x})^{T}\in \mathbb{R}^{n\times m}$.

Sometimes there may not exist such a sufficient subspace rigorously so that $d=m$, or we may want to select less dimensions $d'<d$ for later analysis even in cases where such a subspace exists in order to achieve a more stringent reduction (albeit with a small loss). For convenience, we slightly reformulate the gKDR approach into a more general form without any change to the results in \cite{fukumizu2014gradient}. Let $\mathbf{W}$ be an $m\times m$ matrix with $\mathbf{W}^{T}\mathbf{W}=\mathbf{I}_m$, satisfying $p(Y|\mathbf{X})=\tilde{p}(Y|\mathbf{W}^{T}\mathbf{X})$. In fact, if there exists a $\mathbf{B}$ matrix satisfying $(\ref{prob-eq})$, we can just set $\mathbf{W}=[\mathbf{B}~\mathbf{C}]$, where $\mathbf{C}$ is an
$m\times (m-d)$ matrix such that $\mathbf{C}^{T}\mathbf{C}=\mathbf{I}_{m-d}$ and the column vectors of $\mathbf{C}$ are orthogonal to those of $\mathbf{B}$; otherwise, $\mathbf{W}=\mathbf{B}$ and $d=m$.

Following the same procedure as before, it is easy to see that
\begin{equation}
\mathbf{M}_{ij}(\mathbf{x})=\sum_{a,b=1}^m \mathbf{W}_{ia}\mathbf{W}_{jb}\langle \nabla_a\varphi(\mathbf{W}^{T}\mathbf{x}),
\nabla_b\varphi(\mathbf{W}^{T}\mathbf{x})\rangle_{\mathcal{H}_{\mathcal{Y}}}.
\end{equation}
If there exists $\mathbf{B}$ satisfying $(\ref{prob-eq})$ with $d<m$, $\nabla_a\varphi(\mathbf{W}^{T}\mathbf{x})=0$ for any $a>d$, hence the respective columns  correspond to the zero eigenvalues of $\mathbf{M}(\mathbf{x})$. The projection matrix $\mathbf{W}$ does not depend on the value of $\mathbf{x}$, while the nontrivial eigenvalues vary with $\mathbf{x}$. Therefore, we obtain the following eigen-decomposition
\begin{equation}\label{M-eigen}
\mathbf{M}(\mathbf{x})=\mathbf{W}\boldsymbol{\Lambda}(\mathbf{x})\mathbf{W}^{T},~\boldsymbol{\Lambda}(\mathbf{W})=\diag(\lambda_1(\mathbf{x}),...,\lambda_m(\mathbf{x})).
\end{equation}

\section{Joint emulation with dimension reduction}
\label{sec:model}
The gKDR approach is now applied together with GP emulation, to construct a low-dimensional approximation to a simulator. 
Thus the following procedure is employed to emulate a high-dimensional simulator.

{\bf Step 1.} Given a set of $n_1$ simulator's runs $(\mathbf{X}_1,Y_1),...,(\mathbf{X}_{n_1},Y_{n_1})$, estimate the projection matrix $\tilde{\mathbf{W}}$ using the gKDR approach.

{\bf Step 2.} Split $\tilde{\mathbf{W}}$ into $[\tilde{\mathbf{W}}_1~\tilde{\mathbf{W}}_2]$, where $\tilde{\mathbf{W}}_1$ consists of the first $d$ columns of $\tilde{\mathbf{W}}$ corresponding to the largest $d$ eigenvectors.

{\bf Step 3.} Design a set of $n_2$ runs $(\mathbf{X}'_1,Y'_1),...,(\mathbf{X}'_{n_2},Y'_{n_2})$ of the simulator, e.g. based on the reduced space $\tilde{\mathbf{W}}_1^{T}\mathbf{X}$, and construct an emulator using the lower dimensional pairs $(\tilde{\mathbf{W}}_1^{T}\mathbf{X}'_1,Y'_1),...,(\tilde{\mathbf{W}}_1^{T}\mathbf{X}'_{n_2},Y'_{n_2})$.

In Step 1, sufficient samples are needed to estimate $\tilde{\mathbf{W}}$ accurately. The theoretical results in \cite{fukumizu2014gradient} on the convergence rate of $\tilde{\mathbf{M}}_n$ would provide some insights. In practice, the number of directions that have a major influence may also affect the sample size $n_1$ needed. Step 2 requires an appropriate selection of $d$ to construct an efficient and effective emulator. The samples to train the emulator in Step 3 can be different (e.g. additional runs) from those already collected to find $\tilde{\mathbf{W}}$ in Step 1. There is a benefit in terms of design arising from the dimension reduction. Indeed, in step 3, the design can be built to explore the reduced space of possible $\tilde{\mathbf{W}}_1^{T}\mathbf{X}'$ but the actual inputs of the simulator are of the corresponding high-dimensional values of $\mathbf{X}'$, as the dimensions left out are deemed unimportant. 

\subsection{Approximation properties}
We now explore some theoretical properties of the low-dimensional approximation to a simulator using the gKDR approach. For any $\mathbf{X}=\mathbf{x}\in\mathcal{R}^m$, if $\mathbf{M}(\mathbf{x})$ is known exactly, we have the eigen-decomposition $(\ref{M-eigen})$. Suppose that the eigenvectors and eigenvalues are partitioned as
\begin{equation}
\boldsymbol{\Lambda}(\mathbf{x})=\left[\begin{matrix}
\boldsymbol{\Lambda}_1(\mathbf{x}) & \\
 & \boldsymbol{\Lambda}_2(\mathbf{x})
\end{matrix}\right],~~~\mathbf{W}=[\mathbf{W}_1~\mathbf{W}_2],
\end{equation}
where $\boldsymbol{\Lambda}_1(\mathbf{x})=\diag(\lambda_1(\mathbf{x}),...,\lambda_d(\mathbf{x}))$ with $d<m$ consists of the first $d$ largest eigenvalues, $\mathbf{W}_1$ is the $m\times d$ matrix whose columns are the associated eigenvectors.  Then for any $\mathbf{X}$, we can define the projected coordinates by
$\mathbf{U}=\mathbf{W}_1^{T}\mathbf{X}\in\mathbb{R}^d$ and $\mathbf{V}=\mathbf{W}_2^{T}\mathbf{X}\in\mathbb{R}^{m-d}$. Our proposed approach suggests to make inference on $Y$ based on $\mathbf{U}$ instead of the full explanatory variable $\mathbf{X}$. The following proposition establishes an error bound for such approximation.
\begin{proposition}\label{prop1}
For any $g\in\mathcal{H}_{\mathcal{Y}}$ and $\mathbf{u}\in\mathbb{R}^{d}$, we approximate $\E[g(Y)|\mathbf{X}=\mathbf{x}]$ by $\E[g(Y)|\mathbf{U}=\mathbf{u}]$ for any $\mathbf{x}$ such that $\mathbf{W}_1^{T}\mathbf{x}=\mathbf{u}$. The approximation error is bounded as follows:
\begin{equation}
\left\|\E[g(Y)|\mathbf{X}=\mathbf{x}]-\E[g(Y)|\mathbf{U}=\mathbf{u}]\right\|_{L_2}^2\leq C_1\left(\sum\limits_{i=d+1}^m b_i\lambda_i^2(\mathbf{x})\right),
\end{equation}
where $C_1$ is a constant depending on the domain of $\mathbf{x}$, $b_i~(i=d+1,...,m)$ are positive constants relating to $\mathbf{W}_1$ and $g$.
\end{proposition}
\begin{proof}
Let $G(\mathbf{x})=\E[g(Y)|\mathbf{X}=\mathbf{x}]$, and $\phi_i=C_{Y\mathbf{X}}C_{\mathbf{X}\mathbf{X}}^{-1}\frac{\partial k_{\mathcal{X}(\cdot,\mathbf{x})}}{\partial x_{i}}\in \mathcal{H}_{\mathcal{Y}}$ for $i=1,...,m$. Following \cite{amini2012sampled}, for any $g\in \mathcal{H}_{\mathcal{Y}}$, we can define a bounded linear operator $\Phi: \mathcal{H}_{\mathcal{Y}}\rightarrow \mathbb{R}^m$ on the Hilbert space such that 
\begin{equation}
\Phi g=\left[\langle\phi_1,g\rangle_{\mathcal{H}_{\mathcal{Y}}}~~
\langle\phi_2,g\rangle_{\mathcal{H}_{\mathcal{Y}}}~~\cdots~~
\langle\phi_m,g\rangle_{\mathcal{H}_{\mathcal{Y}}}\right]^{T}.
\end{equation} 
Its adjoint is a mapping $\Phi^{*}: \mathbb{R}^m\rightarrow \mathcal{H}_{\mathcal{Y}}$, defined by the relation $\langle\Phi g,\mathbf{a}\rangle_{\mathbb{R}^m}=\langle g,\Phi^{*}\mathbf{a}\rangle_{\mathcal{H}_{\mathcal{Y}}}$ for any $g\in\mathcal{H}_{\mathcal{Y}}$ and $\mathbf{a}\in\mathbb{R}^m$. Then we have \begin{equation}
\langle\Phi g,\mathbf{a}\rangle_{\mathbb{R}^m}=\sum_{i=1}^m a_i\langle\phi_i,g\rangle_{\mathcal{H}_{\mathcal{Y}}}=\langle\sum_{i=1}^m a_i\phi_i,g\rangle_{\mathcal{H}_{\mathcal{Y}}}=\langle g,\Phi^{*}\mathbf{a}\rangle_{\mathcal{H}_{\mathcal{Y}}}.
\end{equation}
Because $g$ is arbitrage, it must hold for any $\mathbf{a}\in\mathbb{R}^m$ that
$\Phi^{*}\mathbf{a}=\sum\limits_{i=1}^m a_i\phi_i.$

Defining $\mathbf{K}=\Phi\Phi^{*}\in\mathbb{R}^m$, it is easy to see that $\mathbf{K}_{ij}=\langle\phi_i,\phi_j\rangle_{\mathcal{H}_{\mathcal{Y}}}$.
From the derivation of the gKDR approach, the derivative of $G(\mathbf{x})$ w.r.t $\mathbf{x}$ is just $\nabla_{\mathbf{x}} G=\Phi g$ and $\mathbf{M}=\mathbf{K}=\Phi\Phi^{*}$. We denote the range of an operator $A$ as $\Ra(A)$ and its kernel (null space) as $\Ker(A)$. The space $\Ra(\Phi^{*})$ is finite-dimensional and hence closed, so we have the decomposition $\mathcal{H}_{\mathcal{Y}}=\Ra(\Phi^{*})\oplus \Ker(\Phi)$. In particular, for any $g\in\mathcal{H}_{\mathcal{Y}}$, there is $\mathbf{a}\in\mathbb{R}^m$ and $g^{\perp}\in \Ker(\Phi)$ such that $g=\Phi^{*}\mathbf{a}+g^{\perp}$. Hence we obtain 
\begin{equation}
\Phi g=\mathbf{M}\mathbf{a}.
\end{equation}

Given the projection of coordinates from $\mathbf{x}$ to $\mathbf{u}$ and $\mathbf{v}$, we can write
\begin{equation}
G(\mathbf{x})=G(\mathbf{W}\mathbf{W}^{T}\mathbf{x})=
G(\mathbf{W}_1\mathbf{W}_1^{T}\mathbf{x}+\mathbf{W}_2\mathbf{W}_2^{T}\mathbf{x})=
G(\mathbf{W}_1\mathbf{u}+\mathbf{W}_2\mathbf{v}).
\end{equation}
The gradient of $G$ w.r.t $\mathbf{u}$ can be obtained by the chain rule as
\begin{equation}
\nabla_{\mathbf{u}} G=\nabla_{\mathbf{u}} G(\mathbf{W}_1\mathbf{u}+\mathbf{W}_2\mathbf{v})=\mathbf{W}_1^{T}\nabla_{\mathbf{x}} G(\mathbf{x})=\mathbf{W}_1^{T}\mathbf{M}\mathbf{a}=\boldsymbol{\Lambda}_1(\mathbf{x})\mathbf{B}^{T}\mathbf{a},\end{equation}
where $\mathbf{a}\in\mathbb{R}^m$ relates to $g$. Then it is easy to see that
\begin{equation}
\|\nabla_{\mathbf{u}} G\|_{L_2}^2=\sum\limits_{i=1}^d b_i\lambda_i^2(\mathbf{x}),
\end{equation}
where the positive constants $b_i$ depend on $\mathbf{W}_1$ and $g$, for $i=1,...,d$. Similarly, we have
\begin{equation}
\|\nabla_{\mathbf{v}} G\|_{L_2}^2=\sum\limits_{i=d+1}^m b_i\lambda_i^2(\mathbf{x}),
\end{equation}
where the positive constants $b_i$ depend on $\mathbf{W}_2$ and $g$, for $i=d+1,...,m$.

We now infer $g(Y)$ based on $\mathbf{u}\in\mathbb{R}^d$ rather than $\mathbf{x}\in\mathbb{R}^m$ with $d<m$. For any $\mathbf{u}$, we have
\begin{equation}
\E[G|\mathbf{u}]=\int_\mathbf{v} G(\mathbf{W}_1\mathbf{u}+\mathbf{W}_2\mathbf{v})dP(\mathbf{v}|\mathbf{u})=\int_\mathbf{v} \E[g(Y)|\mathbf{u},\mathbf{v}])dP(\mathbf{v}|\mathbf{u})=\E[g(Y)|\mathbf{u}].
\end{equation}
Therefore, for any fixed $\mathbf{u}$, we estimate $G(\mathbf{x})=G(\mathbf{W}_1\mathbf{u}+\mathbf{W}_2\mathbf{v})$ with $E[G|\mathbf{u}]$ for any $\mathbf{x}=\mathbf{W}_1\mathbf{u}+\mathbf{W}_2\mathbf{v}$, i.e.
\begin{equation}
G(\mathbf{x})\approx \hat{G}(\mathbf{x})=\E[G|\mathbf{W}_1^{T}\mathbf{x}]=\E[G|\mathbf{u}].
\end{equation}

Note that for any fixed $\mathbf{u}$, the original function $G(\mathbf{x})=G(\mathbf{W}_1\mathbf{u}+\mathbf{W}_2\mathbf{v})$ is a function of only $\mathbf{v}$, while the approximation $\hat{G}(\mathbf{x})=E\left[G|\mathbf{u}\right]$ is in fact the average of $G(\mathbf{u},\mathbf{v})$ over all possible $\mathbf{v}$, and so is not a function of $\mathbf{v}$. The Poincar\'e inequality yields
\begin{equation}
\|G-\hat{G}\|_{L_2}^2\leq C_1\|\nabla_{\mathbf{v}} G\|_{L_2}^2=C_1\left(\sum\limits_{i=d+1}^m b_i\lambda_i^2(\mathbf{x})\right),
\end{equation}
where $C_1$ is a constant depending on the domain of $\mathbf{x}$.
\end{proof}

When $\mathbf{W}_1$ represents a sufficient dimension reduction, $\lambda_i(\mathbf{x})=0$ for $i=d+1,...,m$, which implies that $\E[g(Y)|\mathbf{X}=\mathbf{x}]=\E[g(Y)|\mathbf{U}=\mathbf{W}_1^{T}\mathbf{x}]$ exactly. Though the result is presented with conditional mean $E[g(Y)|\cdot]$ for any $g\in\mathcal{H}_{\mathcal{Y}}$, it is not limited to the first moment only. For characteristic kernels such as the popular Gaussian RBF kernel $k(x,y)=\exp(-\|x-y\|^2/(2\sigma^2))$ and the Laplace kernel $k(x,y)=\exp(-\alpha\sum_{i=1}^m|x_i-y_i|)$, probabilities are uniquely determined by their means on the associated RKHS \cite{fukumizu2014gradient}; see also \cite{gretton2012kernel} for a definition of the distance between probabilities using their means.

In practice, $\mathbf{W}$ cannot be known exactly. We can only estimate a perturbed version $\tilde{\mathbf{W}}=[\tilde{\mathbf{W}}_1~\tilde{\mathbf{W}}_2]$ instead using the eigen-decomposition of $\tilde{\mathbf{M}}_n$. Under some mild conditions, $\tilde{\mathbf{M}}_n$ converges in probability to $\E[\mathbf{M}(\mathbf{x})]$ with order $O_p\left(n^{-\min\{1/3,(2\beta+1)/(4\beta+4)\}}\right)$ for some $\beta>0$ \cite{fukumizu2014gradient}. As a result, we have the following result.
\begin{proposition}\label{prop2}
For any $g\in\mathcal{H}_{\mathcal{Y}}$ and $\tilde{\mathbf{u}}\in\mathbb{R}^{d}$, we approximate $\E[g(Y)|\mathbf{X}=\mathbf{x}]$ by $\E[g(Y)|\tilde{\mathbf{U}}=\tilde{\mathbf{u}}]$ for every $\mathbf{x}$ such that $\tilde{\mathbf{W}}_1^{T}\mathbf{x}=\tilde{\mathbf{u}}$. Then we have
\begin{align}
\begin{split}
&\left\|\E[g(Y)|\mathbf{X}=\mathbf{x}]-\E[g(Y)|\tilde{\mathbf{U}}=\tilde{\mathbf{u}}]\right\|_{L_2}^2=\\
&O_p\left(\left(\frac{4}{\lambda_d-\lambda_{d+1}}n^{-\min\{\frac{1}{3},\frac{2\beta+1}{4\beta+4}\}}\left(\sum_{i=1}^d b_i\lambda_i^2(\mathbf{x})\right)^{\frac{1}{2}}+\left(\sum_{i=d+1}^m b_i\lambda_i^2(\mathbf{x})\right)^{\frac{1}{2}}\right)^2\right),
\end{split}
\end{align}
where $C_1$ is a constant depending on the domain of $\mathbf{x}$ and the $b_i~(i=1,...,m)$ are positive constants related to $\mathbf{W}$ and $g$.
\end{proposition}
\begin{proof}
Denoting $\tilde{\mathbf{M}}_n=\E[\mathbf{M}(\mathbf{x})]+\mathbf{E}_n$ and $e_n=n^{-\min\{1/3,(2\beta+1)/(4\beta+4)\}}$, the convergence result on $\tilde{\mathbf{M}}_n$ \cite{fukumizu2014gradient} entails that for any $\epsilon>0$, there exists a constant $C>0$ and $N_{\epsilon}$ such that
for any $n\geq N_{\epsilon}$,
\begin{equation}
P\left(\|\mathbf{E}_n\|<Ce_n\right)>1-\epsilon.
\end{equation}
Then there exists $N'_{\epsilon}$ such that for any $n\geq N'_{\epsilon}$,
\begin{equation}
Ce_n\leq \frac{\lambda_d-\lambda_{d+1}}{5},
\end{equation}
where $\lambda_1\geq\lambda_2\geq...\geq\lambda_m\geq 0$ are the eigenvalues of $\E[\mathbf{M}(\mathbf{x})]$; where $N'_{\epsilon}$ can be chosen as $N'_{\epsilon}=\max\left\{N_{\epsilon},\left((\lambda_d-\lambda_{d+1})/(5C)\right)^{-\max\{3,(4\beta+4)/(2\beta+1)\}}\right\}$.

The distance between subspaces that are spanned by columns of $\mathbf{W}_1$ and $\tilde{\mathbf{W}}_1$, denoted by $\spans(\mathbf{W}_1)$ and $\spans(\tilde{\mathbf{W}}_1)$ respectively, can be defined as \cite{golub2012matrix}
\begin{equation}
\dist(\spans(\mathbf{W}_1),\spans(\tilde{\mathbf{W}}_1))=\|\mathbf{W}_1\mathbf{W}_1^{T}-\tilde{\mathbf{W}}_1\tilde{\mathbf{W}}_1^{T}\|=\|\mathbf{W}_1^{T}\tilde{\mathbf{W}}_2\|.
\end{equation}
Using Corollary $8.1.11$ of \cite{golub2012matrix}, we have
\begin{equation}
\|\mathbf{W}_1^{T}\tilde{\mathbf{W}}_2\|\leq 4Ce_n/(\lambda_d-\lambda_{d+1}).
\end{equation}
Hence $\|\mathbf{W}_1^{T}\tilde{\mathbf{W}}_2\|=O_p\left(4e_n/(\lambda_d-\lambda_{d+1})\right).$
We also note that $\|\mathbf{W}_2^{T}\tilde{\mathbf{W}}_2\|\leq \|\mathbf{W}_2\|\|\tilde{\mathbf{W}}_2\|=1$.

Then for any $\mathbf{x}$, we have the following approximation to $G(\mathbf{x})$,
\begin{equation}
G(\mathbf{x})\approx\tilde{G}(\mathbf{x})=\E\left[G|\tilde{\mathbf{W}}_1^{T}\mathbf{x}\right]=\E\left[G|\tilde{\mathbf{u}}\right].
\end{equation}
Letting $\tilde{\mathbf{v}}=\tilde{\mathbf{W}}_2^{T}\mathbf{x}$ and following the same procedure as \Cref{prop1}, for any fixed $\tilde{\mathbf{u}}$ we have,
\begin{equation}
\|G-\tilde{G}\|_{L_2}^2\leq C_1\|\nabla_{\tilde{\mathbf{v}}}G\|_{L_2}^2,
\end{equation}
where $C_1$ is some constant. Since $\nabla_{\tilde{\mathbf{v}}}G = \mathbf{W}_2^{T}\tilde{\mathbf{W}}_2\nabla_{\mathbf{v}} G+\mathbf{W}_1^{T}\tilde{\mathbf{W}}_2\nabla_{\mathbf{u}} G$, we have
\begin{equation}
\|G-\tilde{G}\|_{L_2}^2\leq C_1\|\nabla_{\tilde{\mathbf{v}}}G\|_{L_2}^2
\leq C_1\left(\|\mathbf{W}_2^{T}\tilde{\mathbf{W}}_2\nabla_{\mathbf{v}} G\|_{L^2}+\|\mathbf{W}_1^{T}\tilde{\mathbf{W}}_2\nabla_{\mathbf{u}} G\|_{L_2}\right)^2.
\end{equation}
The result holds by plugging the respective terms.
\end{proof}

The approximation procedure generates an ``innovative simulator'' $\tilde{f}$ on the reduced input space of $\mathbf{U}=\tilde{\mathbf{W}}_1^{T}\mathbf{X}$, which is however not deterministic. Suppose there are two distinct inputs $\mathbf{X}_1$ and $\mathbf{X}_2$ with the respective outputs $Y_1\neq Y_2$. It may happen that $\tilde{\mathbf{W}}_1^{T}\mathbf{X}_1=\tilde{\mathbf{W}}_1^{T}\mathbf{X}_2$, i.e. the approximated simulator $\tilde{f}$ may yield different outputs given the same input. The low-dimensional stochastic simulator $\tilde{f}$ can nevertheless be emulated, for example using GP with nugget effect, assuming that the influence of the dropped components is relatively small and simple enough to be captured by the nugget. The overall approximate error of the final emulator $\hat{f}$ to $f$ can be decomposed into $\|\hat{f}-f\|\leq\|f-\tilde{f}\|+\|\tilde{f}-\hat{f}\|$, where the first term in the right hand side is due to the low-dimensional approximation which has been investigated in \cref{prop2}, and the second term depends on the emulation procedure. 

\subsection{Choice of parameters and structural dimension}
When applying the proposed framework for emulation, several parameters need to be specified properly, e.g. the parameters in the kernels and the regularization parameter $\epsilon_n$. The cross validation approach can be used to tune such parameters as in many nonparametric statistical methods. In addition, it is also required to choose an appropriate structural dimension $d$ to construct an accurate emulator. 

One of the possible ways is to choose $d$ within the dimension reduction procedure. Fukumizu and Leng \cite{fukumizu2014gradient} pointed out that it might not be practical to select $d$ based on asymptotic analysis of some test statistics, as in many existing dimension reduction techniques, when the dimension is high and the sample size is small. They mentioned that the ratio of the sum of the largest $d$ eigenvalues over the sum of all the eigenvalues, $\sum_{i=1}^d\lambda_i/\sum_{i=1}^m\lambda_i$, might be useful in identifying the conditional independence of $Y$ and $\mathbf{X}$ given $\mathbf{B}^{T}\mathbf{X}$. In addition, \cref{prop2} shows that the approximation error decreases as a function of $\lambda_d-\lambda_{d+1}$. As discussed in \cite{constantine2015computing}, $d$ might be chosen such that $\lambda_d-\lambda_{d+1}$ is maximized. However, we may notice that the approximation error also depends on the squares of the eigenvalues with some unknown weights. Therefore it seems to be not very practical to select $d$ solely based on the eigenvalues.

On the other hand, Fukumizu and Leng \cite{fukumizu2014gradient} suggested to select $d$ based on the subsequent utilization of $d$ rather than the dimension reduction procedure when dimension reduction serves as a pre-processing step. For example the ultimate goal of our proposed framework here is to construct an accurate emulator, hence it is intuitive to select the structural dimension that produces the best predictive performance. Therefore, in the following numerical studies, we select $d$ as well as other parameters for the gKDR approach using simple trial-and-error or more formal cross validation approach based on the predictive accuracy of the respective emulators.

\section{Numerical simulations}
\label{sec:sim}
In this section we conduct two numerical studies. In the first study, the proposed emulation framework using the gKDR approach is compared with several alternatives of dimension reduction and the full emulation on a PDE problem. This problem set up allows the computation of gradients explicitly. In the second study, we illustrate the emulation framework with an application to tsunami modeling; we also provide a comparison to other methods, except AS which cannot be applied. Throughout the simulations, the GPML code using maximum likelihood method implemented by \cite{rasmussen2006gaussian} is employed for the emulation assuming a linear form mean function with intercept, and a squared exponential correlation function.

\subsection{Study 1: elliptic PDE with explicit gradients available}
\label{subsec:study1}

In this example, we investigate the elliptic PDE problem with random coefficients as studied in \cite{constantine2014active}. Let $u=u(\mathbf{s},\mathbf{x})$ satisfy the linear elliptic PDE \begin{equation*}-\nabla_{\mathbf{s}}\cdot(a\nabla_{\mathbf{s}}u)=1,~~\mathbf{s}\in [0,1]^2.\end{equation*}
The homogeneous Dirichlet boundary conditions are set on the left, top and bottom boundary (denoted by $\Gamma_1$) of the spatial domain of $\mathbf{s}$, and a homogeneous Neumann boundary condition is imposed on the right side of the spatial domain denoted $\Gamma_2$. The coefficients $a=a(\mathbf{s},\mathbf{x})$ are modeled by a truncated Karhunen-Loeve (KL) type expansion
\begin{equation}
\log(a(\mathbf{s},\mathbf{x}))=\sum_{i=1}^m x_i\gamma_i\phi_i(\mathbf{s}),
\end{equation}
where the $x_i$ are i.i.d. standard Normal random variables, and ${\phi_i(\mathbf{s}),\gamma_i}$ are the eigenpairs of the correlation operator
\begin{equation}
C(\mathbf{s},\mathbf{t})=\exp(\beta^{-1}\|\mathbf{s}-\mathbf{t}\|_{1}).
\end{equation}

The target value is a linear function of the solution
\begin{equation}
f(\mathbf{x})=\int_{\Gamma_2}u(\mathbf{s},\mathbf{x})/|\Gamma_2|d\mathbf{s}.
\end{equation}
The problem is discretized using a finite element method on a triangulation mesh, then $f$ and $\nabla_{\mathbf{x}}f$ can be computed as a forward and adjoint 
problem; see \cite{constantine2014active} for more details. We choose $m=100$ and examine two cases of the correlation lengths $\beta=1$ or $\beta=0.01$. Therefore, the original input space is $\mathcal{X}=\mathbb{R}^{100}$ with standard Normal distribution and the output $f(\mathbf{x})$ is univariate.

The gKDR approach is applied to reduce the dimension of the problem using $M$ samples. We also compare with several popular alternative dimension reduction techniques: AS (here possible due to the explicit gradients), SIR, SIR-II \cite{li1991sliced}, sliced average variance estimation (SAVE) \cite{cook1991discussion}, MAVE and PFC. After reducing the dimension of the problem, the GP emulator is trained using a Latin Hypercube design of $10d$ points on the reduced $d$-dimensional space so that the whole procedure needs $M+10d$ samples in total using each dimension reduction method. For comparison, we also emulate the problem on the original $100$-dimensional input space directly with $M+10d$ samples, which is the full emulation. The gKDR approach is implemented in Matlab by K. Fukumizu, see \url{http://www.ism.ac.jp/\string~fukumizu/}. The Matlab code for AS and solving the PDE by \cite{constantine2014active} is available on \url{https://bitbucket.org/paulcon/active-subspace-methods-in-theory-and-practice}. For SIR, SAVE and PFC, the codes are provided in the Matlab LDR-package (\url{https://sites.google.com/site/lilianaforzani/ldr-package}), and SIR-II is implemented by simply modifying the SIR code. For MAVE, the Matlab code by Y. Xia is available from \url{http://www.stat.nus.edu.sg/\string~staxyc/}. The associated parameters in some methods, such as the kernel and regularization parameters for gKDR, the number of slices for the sliced methods and the degree of polynomial basis for PFC, are chosen in a simple trial-and-error way by trying several values and selecting the best.

The final emulators are used to make prediction on a testing set of $n$ evaluations $\{f_1,...,f_n\}$ that differ from the training set, where $f_i=f(\mathbf{x}_i)$ and $\mathbf{x}_i\in\mathbb{R}^{100}$ is drawn randomly from the standard Normal distribution. The predictive performance is measured by the normalized predictive root-mean-square-error (PRMSE)
\begin{equation*}\Nprmse=\frac{\sqrt{\frac{1}{n}\sum_{i=1}^n (f_i-\hat{f}_i)^2}}{\max_i f_i - \min_i f_i},
\end{equation*}
where $\hat{f}_i$ is the prediction (predictive mean) using emulation. The associated computing time is also recorded with three parts: T1 for running the simulator, T2 for estimating the dimension reduction and T3 for training the emulator and making prediction. Note that T1 includes the time devoted to run the simulator $M$ times when using all the dimension reduction methods. It also includes the time used to compute the gradients for AS and the additional $10d$ runs for full emulation. T2 is zero for full emulation since there is no dimension reduction. T3 also includes the time for running the simulator $10d$ times on the designed points except full emulation.

In this study, we choose $M=300$ and $d=1,...,5$. \Cref{study1_table} presents the results on a testing set with $n=500$ evaluations using different emulation approaches and \cref{dr_emu_time_records} shows an example of the associated computing time when $\beta=1$ and $d=5$. Compared with the full emulation results, by reducing the dimension properly, the predictive accuracy can be improved, especially when the correlation length is long ($\beta=1$). Also, as a result, the computing time for training GP emulator (T3) decreases dramatically. In terms of predictive accuracy, AS naturally performs the best, as it is using exact gradients, followed by gKDR. MAVE, SIR and PFC are better than SIR-II and SAVE, but PFC does not work very well when $\beta=1$ and MAVE spends more computing time on dimension reduction. Most methods yield smaller errors for $\beta=1$ than $\beta=0.01$, except PFC and full emulation. Unlike the other techniques, AS employs exact gradients which explains its lead in performance. However, as shown in \cref{dr_emu_time_records}, the computing time T1 for AS is about two orders of magnitude longer than the others making the method most computationally expensive. Moreover, computing gradients $\nabla_{\mathbf{x}} f$ is sometimes impossible, e.g. for the tsunami simulation in the next study. This restricts the applicability of AS method to a few applications. To summarize, when the exact gradients are computable the proposed gKDR approach is able to produce comparable results (though not as good) as the AS method that uses exact gradients, and outperforms the other SDR methods in most cases. However, the computational cost of applying gKDR is much less than that employing AS. In fact, gKDR not only is able to find the SDR accurately and efficiently, but also can be applied in a wide range of scenarios where complicated variable types or very high dimensions are involved. The next application into tsunami simulation provides a snapshot of its wide capability when there are few applicable alternatives.

\begin{table}[htp!]
\centering
\caption{Study 1: Normalized PRMSEs at $500$ testing sites using emulation on the full input space (Full) or combined with different dimension reduction techniques.}
\begin{tabular}{|c|c|c|c|c|c|c|c|c|}\hline 
\multicolumn{9}{|c|}{$\beta=1$} \\ \hline
 d & gKDR & AS & SIR & SIR-II & SAVE & MAVE & PFC & Full \\ \hline
1 & $0.116$ & $0.126$ & $0.125$ & $0.153$ & $0.153$ & $0.126$ & $0.152$ & $0.097$ \\ \hline
2 & $0.044$ & $0.007$ & $0.025$ & $0.153$ & $0.153$ & $0.020$ & $0.140$ & $0.095$ \\ \hline
3 & $0.032$ & $0.011$ & $0.024$ & $0.152$ & $0.152$ & $0.019$ & $0.120$ & $0.095$ \\ \hline
4 & $0.024$ & $0.012$ & $0.024$ & $0.150$ & $0.150$ & $0.024$ & $0.080$ & $0.093$ \\ \hline
5 & $0.024$ & $0.011$ & $0.026$ & $0.150$ & $0.150$ & $0.083$ & $0.071$ & $0.092$ \\ \hline
\multicolumn{9}{|c|}{$\beta=0.01$} \\ \hline
1 & $0.037$ & $0.033$ & $0.043$ & $0.169$ & $0.169$ & $0.039$ & $0.161$ & $0.032$ \\ \hline
2 & $0.033$ & $0.028$ & $0.039$ & $0.169$ & $0.169$ & $0.038$ & $0.160$ & $0.032$ \\ \hline
3 & $0.033$ & $0.029$ & $0.039$ & $0.167$ & $0.167$ & $0.039$ & $0.034$ & $0.032$ \\ \hline
4 & $0.033$ & $0.025$ & $0.039$ & $0.167$ & $0.167$ & $0.039$ & $0.033$ & $0.032$ \\ \hline
5 & $0.033$ & $0.024$ & $0.038$ & $0.167$ & $0.167$ & $0.037$ & $0.033$ & $0.032$ \\ \hline
\end{tabular}
\label{study1_table}
\end{table}

\begin{figure}[htbp!]
\centering
\includegraphics[width=0.8\textwidth]{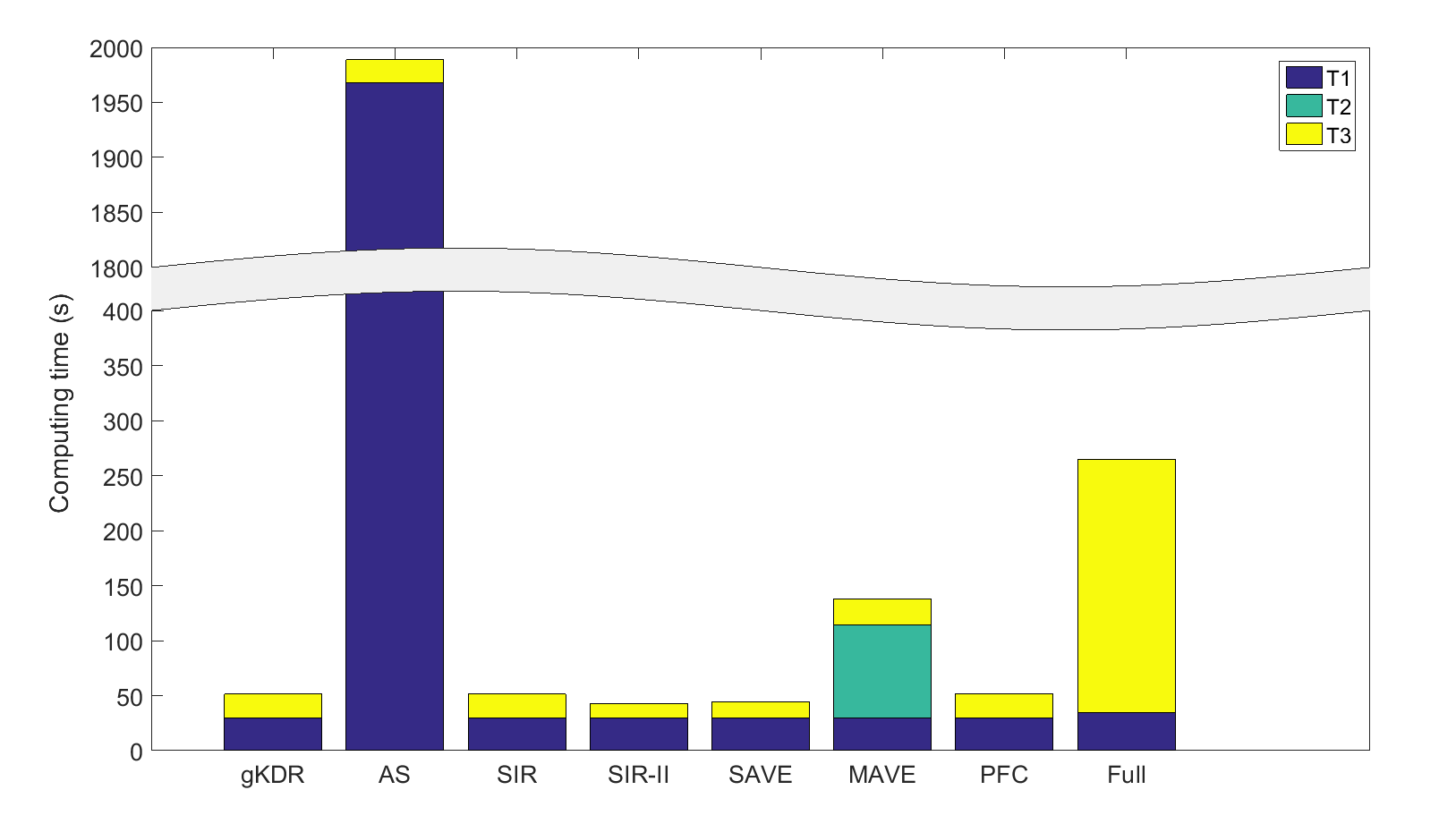}
\caption{Computing time (in seconds) for the emulation using different approaches when $\beta=1$ and $d=5$.}
\label{dr_emu_time_records}
\end{figure}

\subsection{Study 2: tsunami emulation where no gradients available}
\label{subsec:study2}
Here we apply the proposed general framework to investigate the impact of uncertainties in the bathymetry on tsunami modeling, where the bathymetry is included as a high-dimensional input.  

A synthetic bathymetry surface is created in the $(s_1,s_2)$ coordinate system to conduct tsunami simulations as shown in \cref{syn_bathy_seadeform_gauges} (a). For simplicity, we assume that the seabed elevation only vary along the first coordinate $s_1$. Though simple, it still captures the typical continental characteristics: the continental shelf spans from shore line ($s_1=0$) to around $s_1=-25$ km at the water depth of around $150$ m; the continental slope is between $s_1=-25$ km and $s_1=-75$ km with water depth of $150\sim1800$ m; west of $-75$ km it is the deep ocean with water depth of $1800\sim2200$ m .

\begin{figure}[htbp!]
\centering
\includegraphics[width=\textwidth]{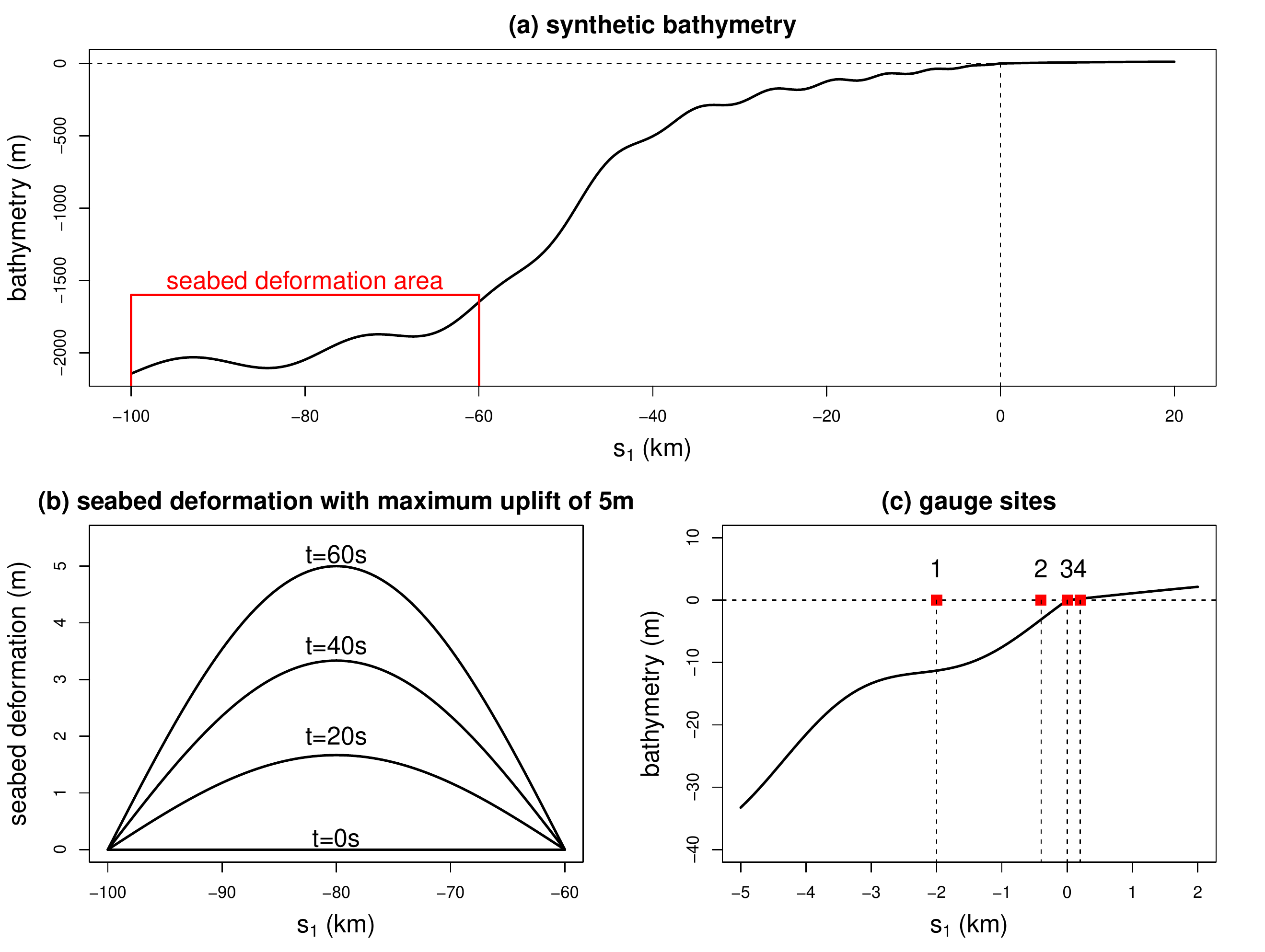}
\caption{(a) Synthetic bathymetry; (b) seabed uplift when $h_{max}=5$ m; (c) gauge sites.}
\label{syn_bathy_seadeform_gauges}
\end{figure}

To model uncertainties and mimic the realistic boat tracks of oceanic surveys, some irregular lines are drawn. We consider two levels of survey density which are denoted by survey level 1 and 2 respectively. Considering that the surveys are usually constrained within budgets, the total lengths of the two level surveys are fixed at $1000$ and $100$ km. To account for different possible survey traces, $20$ samples of boat tracks are drawn at each level of survey density; see in \cref{boat_track_samples} three samples per level for illustration. In this study, we only consider the impact of the uncertainties in the bathymetry within the area $(s_1,s_2)\in[-40000,0]\times[-5000,5000]$ as shown with a blue rectangle in \cref{boat_track_samples}. The bathymetry at other locations are fixed at the true values. This assumption is based on the physical knowledge that deep ocean has a relatively small influence on tsunami waves.

\begin{figure}[htbp!]
\centering
\includegraphics[width=\textwidth]{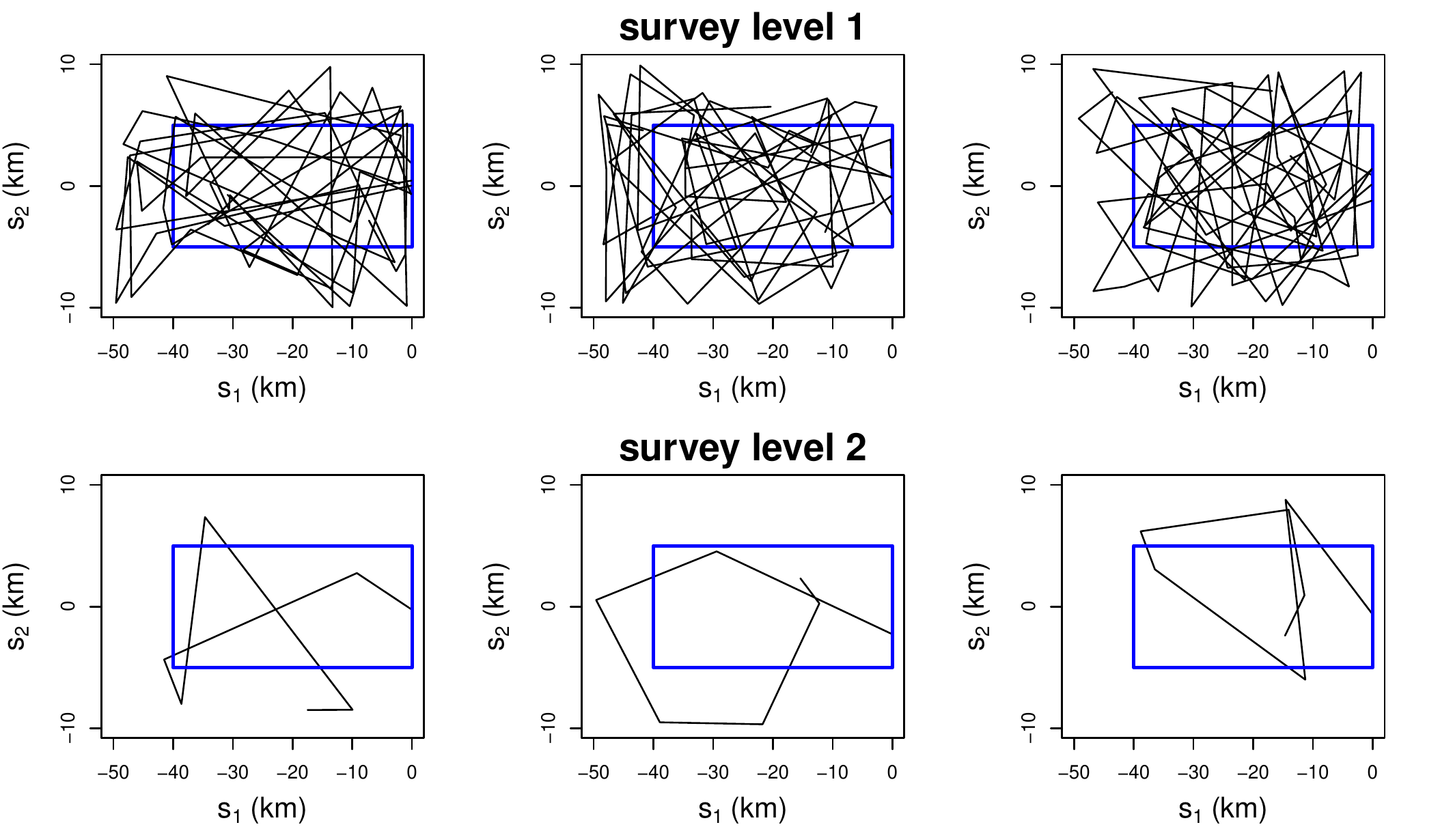}
\caption{Three samples of boat tracks at two levels of survey density; the bathymetry within the blue rectangle are assumed uncertain.}
\label{boat_track_samples}
\end{figure}

Along the possible boat tracks, observations of bathymetry are collected every $500$ m. Then the whole bathymetry surface can be modeled using the SPDE approach \cite{lindgren2011explicit} and inferred using INLA method\cite{rue2009approximate} for approximate Bayesian inference. Given observations of bathymetry $\mathbf{z}=(z_1,...,z_n)'$ at locations $\mathbf{s}=(\mathbf{s}_1,...,\mathbf{s}_n)'$, it is assumed that $z_i=Z(\mathbf{s}_i)+\epsilon_i$, $i=1,...,n$, where the unknown bathymetry surface $Z(\mathbf{s})$ is Gaussian field with Mat\'ern covariance function
\begin{equation}
\operatorname{Cov}(Z(\mathbf{s}),Z(\mathbf{s}^{*}))=\frac{\sigma^2}{2^{\nu-1}\Gamma(\nu)}(\kappa\|\mathbf{s}-\mathbf{s}^{*}\|)^\nu K_{\nu}(\kappa\|\mathbf{s}-\mathbf{s}^{*}\|),
\end{equation}
where $\|\mathbf{s}-\mathbf{s}^{*}\|$ is the Euclidean distance between two locations $\mathbf{s}$ and $\mathbf{s}^{*}\in\mathbb{R}^2$, $K_{\nu}$ is the modified Bessel function of the second kind and order $\nu > 0$, $\kappa > 0$ controls the nominal correlation range through $\rho=\sqrt{8\nu}/\kappa$ corresponding to correlations near $0.1$ at the Euclidean distance $\rho$, and $\sigma^2$ is the marginal variance. Lindgren et al. \cite{lindgren2011explicit} noted that $Z(\mathbf{s})$ also satisfies the stochastic partial differential equation (SPDE)
\begin{equation}
\tau(\kappa^2-\Delta)^{\alpha/2}Z(\mathbf{s})=W(s),
\end{equation}
where the innovation process $W$ is spatial Gaussian white noise with unit variance, $\Delta = \frac{\partial^2}{\partial s_1^2}+\frac{\partial^2}{\partial s_2^2}$ is the Laplacian operator, and $\tau$ controls the marginal variance through the relationship
\begin{equation}
\tau^2=\frac{\Gamma(\nu)}{\Gamma(\nu+1)(4\pi)\kappa^{2\nu}\sigma^2}.
\end{equation}
With a finite elements representation
\begin{equation}
Z(\mathbf{s})=\sum_{k=1}^m w_k\psi_k(\mathbf{s}),
\end{equation}
over an appropriate triangular mesh, a stochastic weak solution to the SPDE can be approximated. It is shown that the coefficients $\mathbf{w}=(w_1,...,w_m)^T$ can be approximated by a Gaussian Markov random field, i.e. $\mathbf{w}\sim \N(\mathbf{0},\mathbf{Q}^{-1})$ for $\mathbf{Q}$ is sparse. Note that bivariate splines could be used \cite{liu2015efficient} to reduce the number of parameters required for specific approximation order, which is good, but not enough, for dimension reduction. Then we build the following hierarchical spatial model, 
\begin{align*}
\begin{split}
\mathbf{z}|\mathbf{w},\boldsymbol{\theta} &\sim \N(\mathbf{Aw},\sigma_e^2\mathbf{I}),\\
\mathbf{w}|\boldsymbol{\theta} &\sim \N(\mathbf{0},\mathbf{Q}(\boldsymbol{\theta})^{-1}),\\
\boldsymbol{\theta} &\sim \pi(\boldsymbol{\theta}),
\end{split}
\end{align*}
where $\mathbf{A}_{ij}=\psi_j(\mathbf{s}_i)$ and $\boldsymbol{\theta}$ contains all the hyperparameters. Since $\mathbf{w}$ uniquely determines the bathymetry, it is the \textit{de facto} input for uncertain bathymetry. In this study, we build a mesh for the finite elements representation in the SPDE approach as shown in \cref{meshes} (a). The dense triangles in the middle cover the uncertain bathymetry area and the outer extension with coarse triangles is added to avoid boundary effect. There are $3200$ nodes that influence the bathymetry, hence the uncertain input for bathymetry is of dimension $3200$. Given each boat track and the associated observations $\mathbf{z}$, $20$ samples of the finite element coefficients are drawn from the posterior $\pi(\mathbf{w}|\mathbf{z})$ to construct a range of possible initial bathymetries. Thus there are $400$ ($20$ samples of boat tracks times $20$ samples of $\mathbf{w}$) sets of possible initial bathymetries in total at each survey level. \Cref{bathy_msd} shows the empirical sample mean and standard deviation of these $400$ possible bathymetries at survey 1 and 2.

\begin{figure}[htbp!]
\centering
\includegraphics[width=\textwidth]{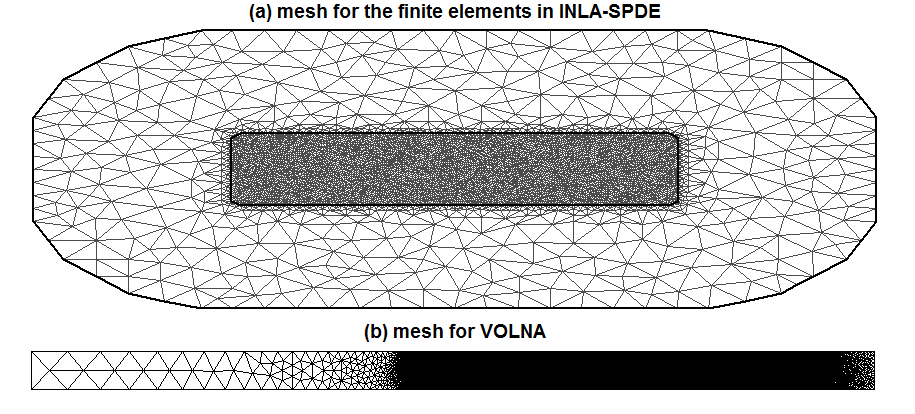}
\caption{(a) mesh for SPDE approach; (b) mesh for VOLNA.}
\label{meshes}
\end{figure}

\begin{figure}[htbp!]
\centering
\includegraphics[width=\textwidth]{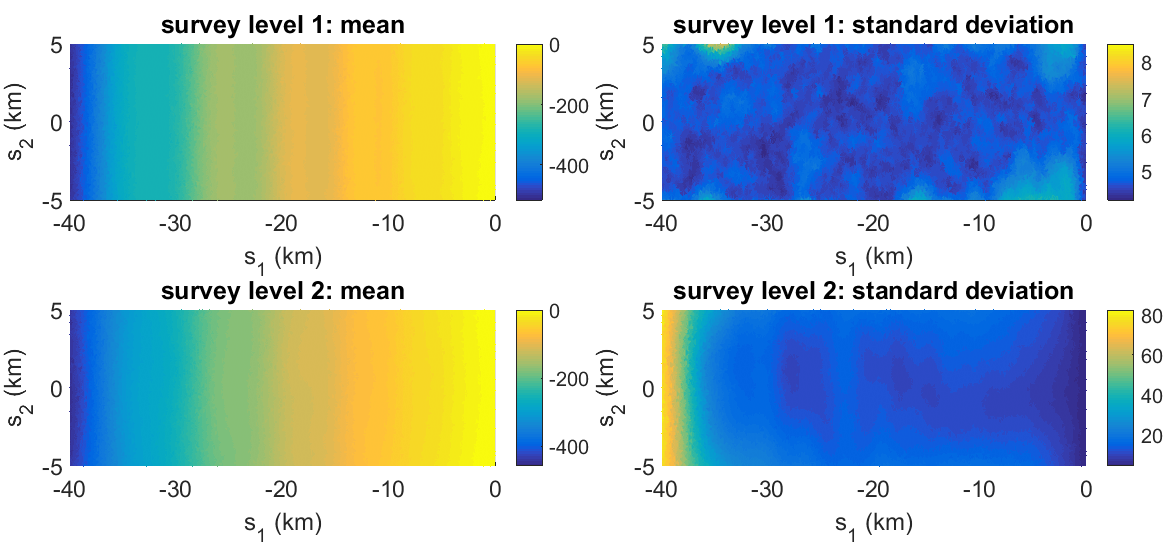}
\caption{Empirical mean and standard deviation of the $400$ sets of bathymetry input across both uncertain boat tracks and posterior samples of $\mathbf{w}$; note the different scales of standard deviation for survey level 1 and 2.}
\label{bathy_msd}
\end{figure}

Tsunami waves are triggered by the following simplified seabed deformation,
\begin{equation}dz(s_1,s_2;t)=\frac{t}{60}\cdot h_{max}\cdot \sin\left(\frac{s_1+100000}{-60000+100000}\pi\right)\cdot I_{\{-100000\leq s_1\leq -60000,0\leq t\leq 60\}},
\end{equation}
where $dz(s_1,s_2;t)$ is the seabed uplift at location $\mathbf{s}=(s_1,s_2)$ and time $t$, $h_{max}$ denotes the maximum seabed uplift; see \cref{syn_bathy_seadeform_gauges} (b) for example. We take $5$ different values $h_{max}=1,...,5$ m. These values are evenly combined with the uncertain initial bathymetry.  Thus there are two sources of uncertainties: $\mathbf{w}$ for bathymetry and $h_{max}$ for tsunami source, where $\mathbf{w}$ is high-dimensional.

We employ the tsunami code VOLNA \cite{dutykh2011volna}, an advanced non-linear shallow water equation solver using the finite volume method on a high performance computing facility. The computational domain and mesh for VOLNA are presented in \cref{meshes} (b). There are $120,661$ triangles and $61,068$ nodes in the mesh, where the coarse triangles in both ends are added to avoid boundary reflection. The output of the simulation is chosen to be $\Delta\eta(\mathbf{s})=\max\eta_t(\mathbf{s})-\eta_0(\mathbf{s})$, where $\eta_t(\mathbf{s})$ is the free surface elevation at simulation time $t$ and location $\mathbf{s}$. $\Delta\eta$ represents the maximum wave height at off shore locations or the maximum inundation depth at on shore locations. For illustration, we consider simulation values at gauge 1: $(-2000,0)$, gauge 2: $(-400,0)$, gauge 3: $(0,0)$ and gauge 4: $(200,0)$, which are at far shore, near 
shore, shore line and land respectively; see \cref{syn_bathy_seadeform_gauges} (c).

The simulation results are presented in \cref{gauges_eta_info}, with those using the true bathymetry shown in red lines and those using the sample mean bathymetry shown in green dash lines. We can see that $\Delta\eta$ increases with $h_{max}$ but also shows variation due to the uncertain inputs $\mathbf{w}$ for fixed $h_{max}$, especially at gauges 2-4 around the shore line. In general, the simulations with sample mean bathymetry would deviate from true values, while those with random bathymetry samples can cover the true events quite well. The survey level also has significant influence  that differs across the four gauges. In most cases the wider range of possible simulation values with coarser survey level 2 indicates that the uncertainty in the bathymetry would spread the tsunami waves out to simulate more extreme scenarios and such effect could be amplified around the shore line.

\begin{figure}[htbp!]
\centering
\includegraphics[width=\textwidth]{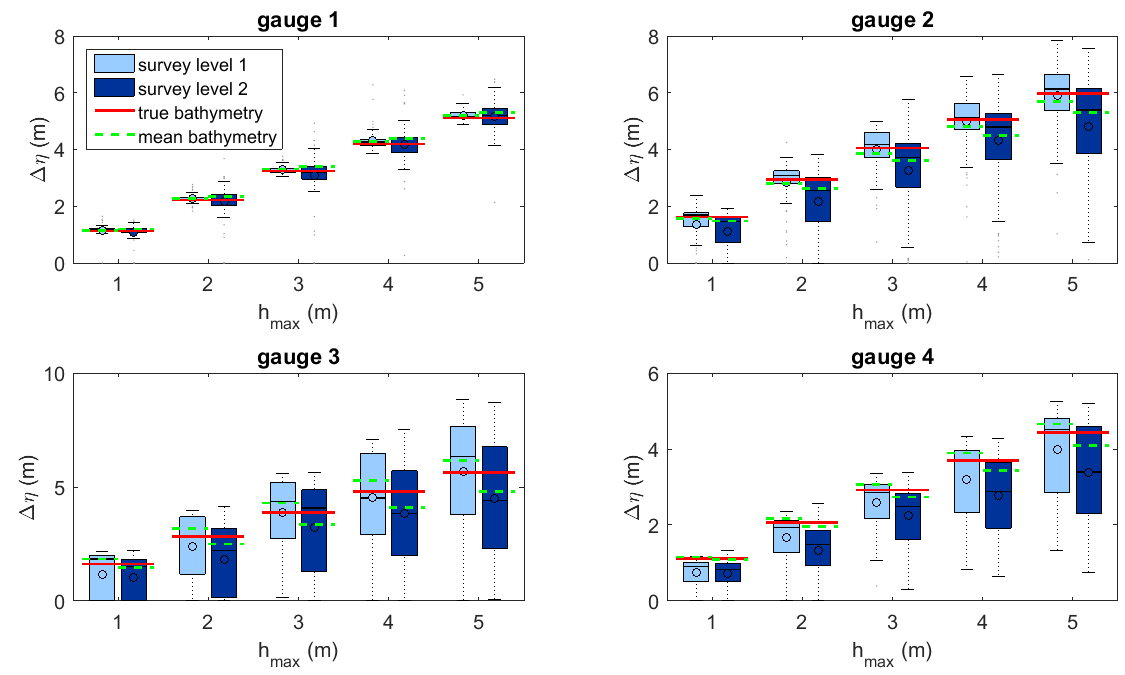}
\caption{Simulation values with different inputs $(\mathbf{w},h_{max})$ at four gauges. The uncertain input $\mathbf{w}$ are drawn based on survey level 1 and 2, together with the true values and sample mean values.}
\label{gauges_eta_info}
\end{figure}

Following the procedure in \cref{sec:model}, we can construct a low-dimensional emulator for such high dimensional simulator with $3200$ input parameters for the bathymetry ($\mathbf{w}$) and $1$ parameter for the seabed deformation ($h_{max}$). Denoting the VOLNA code with $f$, the output can be represented as $\Delta\eta=f(\mathbf{w},h_{max}).$
Because \cref{gauges_eta_info} displays a significant relationship between $h_{max}$ and $\Delta\eta$, we keep it as a separate input in the emulator and reduce the dimension of $\mathbf{w}$ only. In this case, we try to find a projection matrix $\mathbf{B}$ such that $\mathbf{w}\perp (\Delta\eta,h_{max})|\mathbf{B}^{T}\mathbf{w}$. The conditional independence just implies the sufficiency of $\mathbf{B}^{T}\mathbf{w}$, i.e. $p(\Delta\eta|h_{max},\mathbf{w})=\tilde{p}(\Delta\eta|h_{max},\mathbf{B}^{T}\mathbf{w})$ \cite{yin2015sequential}. Therefore, $(\Delta\eta,h_{max})$ is regarded as a temporary output when applying gKDR to reduce the dimension of $\mathbf{w}$.

For the gKDR approach, Gaussian RBF kernels $k(x,y)=\exp(\|x-y\|^2/(2\sigma^2))$ are deployed for both $k_{\mathcal{X}}$ and $k_{\mathcal{Y}}$ but with different parameters $\sigma^2$. Following \cite{fukumizu2014gradient}, we have the parameterization $\sigma_{\mathcal{X}}=c_1\sigma_{med}(\mathbf{X})$, $\sigma_Y=c_2\sigma_{med}(\mathbf{Y})$ where $\sigma_{med}(\cdot)$ is the median of pairwise distance of the data. The regularization parameter $\epsilon_n$ is fixed at $10^{-5}$ because its influence is shown to be negligible after a few trials. There are three parameters to be specified properly: $c_1$, $c_2$, and $d$, the number of directions included in the emulator. We consider possible candidates $c_1,c_2\in [0.5,1,5,10,15,20]$, $d\in[1,...,5]$ here. Then, based on each of the possible parameters combination, we can construct a GP emulator on the low-dimensional inputs $(h_{max},\mathbf{B}^{T}\mathbf{w})$ and make predictions on the new inputs $(\tilde{h}_{max},\mathbf{B}^{T}\tilde{\mathbf{w}})$. 

For comparison, we also apply alternative dimension reduction techniques to construct the low-dimensional approximations. Due to the complexity of the VOLNA code, the gradients of simulation values with respect to the inputs are not computable. Hence the active subspace method cannot be employed. Most of the methods in Study 1 cannot be applied directly because of the need for partial dimension reduction, or the ``large $p$, small $n$'' feature, i.e. there are much more input parameters than the number of simulations. We consider two extensions to PFC and SIR. The partial PFC (PPFC) method  \cite{kim2011partial} is implemented based on the R package ldr \cite{adragni2014ldr} to find the reduction on $\mathbf{w}$ only meanwhile taking the effect of $h_{max}$ into account. Note that PFC is not developed for the problem where $p>n$ or $p\gg n$. Another method we compare to is the sequential sufficient dimension reduction (SSDR) \cite{yin2015sequential}. It is specifically proposed to overcome the ``large $p$, small $n$'' difficulty by decomposing the variables into pieces each of which has $p_1<n$ variables so that conventional dimension reduction methods can be applied. The projective resampling approach \cite{li2008projective} with SIR is employed. The R code for SSDR by  \cite{yin2015sequential} is available from \url{http://wileyonlinelibrary.com/journal/rss-datasets}.

To measure the predictive performance and select proper parameters, a $10$-fold cross validation approach is employed. For each survey level, $400$ simulations are divided evenly into $10$ groups. Each group is retained as testing set once, while the other nine groups are used to estimate the projection matrix $\mathbf{B}$ using the gKDR, PPFC or SSDR approaches and train the respective GP emulator. \Cref{10foldcv_table} presents the normalized PRMSEs from the cross validation for each survey level and gauge using different dimension reduction techniques. The errors of survey level 1 are in general smaller than those of survey level 2. This implies that as the uncertainties in the bathymetry increase, it gets more difficult to make accurate predictions using emulation. The methods gKDR and SSDR outperform PPFC in all cases, especially in survey level 1 where the normalized PRMSEs can be $50\%$ lower in some cases. In survey level 1, the errors of the gKDR approach are slightly larger than those of the SSDR for gauge 2-4 where the normalized PRMSEs using SSDR are around $1.1\%\sim 3.7\%$ lower. But in survey level 2, the gKDR approach is more accurate than the SSDR approach for all gauges with reduction of normalized PRMSEs at $1.0\%$ for gauge 1 and $10.1\%\sim 18.9\%$ for gauge 2-4. Therefore, gKDR is comparable with SSDR in survey level 1 but works much better than SSDR in survey level 2 when there are more uncertainties involved. We can conclude that the proposed GP emulation framework combined with the gKDR dimension reduction approach is effective and accurate for this complicated tsunami simulator and overall it outperforms the alternatives.

\begin{table}[htbp]
\centering
\caption{Study 2: Normalized PRMSEs of the 10-fold cross validation using GP emulation combined with different dimension reduction methods.}
\begin{tabular}{|c|c|c|c|c|c|c|}\hline
\multirow{2}{*}{gauge}  & \multicolumn{3}{|c|}{survey level 1} &  \multicolumn{3}{|c|}{survey level 2} \\ \cline{2-7}
 & gKDR & PPFC & SSDR & gKDR & PPFC & SSDR \\ \hline
1  & $0.031$ & $0.078$ & $0.033$ & $0.095$ & $0.096$ & $0.096$ \\ \hline
2  & $0.099$ & $0.138$ & $0.097$ & $0.134$ & $0.175$ & $0.149$ \\ \hline
3  & $0.091$ & $0.187$ & $0.090$ & $0.129$ & $0.210$ & $0.159$ \\ \hline
4  & $0.082$ & $0.144$ & $0.079$ & $0.106$ & $0.141$ & $0.121$ \\ \hline
\end{tabular}\\
\label{10foldcv_table}
\end{table}

To investigate the impact of the training set size on the predictive performance of the proposed emulation framework with gKDR, we conduct repeated random sub-sampling cross validations with various training set sizes. We consider training set size as $2\%$, $5\%$, $10\%$, $20\%$, ..., $90\%$ and test set size as $10\%$ of the total $400$ simulations. For each training set size, the sampling is repeated $50$ times. The parameters $c_1,c_2,d$ are fixed at those values selected through the above 10-fold cross validation. \Cref{rrsscv_res} displays the normalized PRMSEs with various training set sizes. In general, the predictive errors decrease as the training set size increases, and eventually converge to a relatively flat level, after about $100$ simulations. It is reassuring that such a small number of simulations is enough to allow an efficient and effective dimension reduction and Gaussian process emulation.

\begin{figure}[htbp]
\centering
\includegraphics[width=\textwidth]{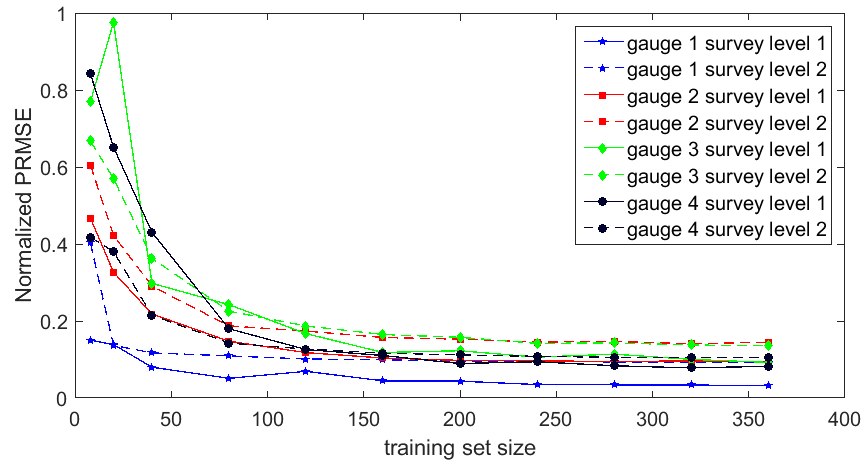}
\caption{Normalized PRMSEs with various training set sizes.}
\label{rrsscv_res}
\end{figure}

In the end, we apply the resulting emulator to predict the simulation values over a large number of new inputs. The predictions can be used for probabilistic risk assessment and many other purposes. For illustration, $10,000$ samples of $(\tilde{h}_{max},\tilde{\mathbf{w}})$ are drawn where $\tilde{h}_{max}$ are drawn from a Normal distribution $\N(3,1)$ truncated at $0$ and $5$. For each survey level, $100$ samples of possible boat tracks are drawn randomly. Given the observations along each boat track, $100$ samples of $\tilde{\mathbf{w}}$ are drawn from the posterior. In most of the current tsunami research work, the bathymetry is usually considered as fixed, which would neglect the possible uncertainty in the outputs. For comparison, we also conduct another set of simulations with the $5$ possible values of $h_{max}$, but with a fixed $\mathbf{w}$ taken to be the sample mean shown in \cref{bathy_msd}. In this case, $h_{max}$ is the only uncertain parameter. Then we can also make another set 
of predictions on the $10,000$ samples of $\tilde{h}_{max}$ only. The predictions for these two cases are presented in \cref{hist_pred}. We can see that at gauge 1 it makes no significant difference to include the uncertainty in the bathymetry or not, as the impact is relatively small on the far shore waves, as shown in \cref{gauges_eta_info}. However, the impact of the uncertainty in the bathymetry on the simulation values is more significant at gauges around the shore line. The distributions are shifted, skewed and spread out, covering more extreme events with larger $\Delta\eta$. These features are potentially important, for example in the catastrophe models for (re)insurance, or hazard assessment used in coastal planning.

\begin{figure}[htbp!]
\centering
\includegraphics[width=\textwidth]{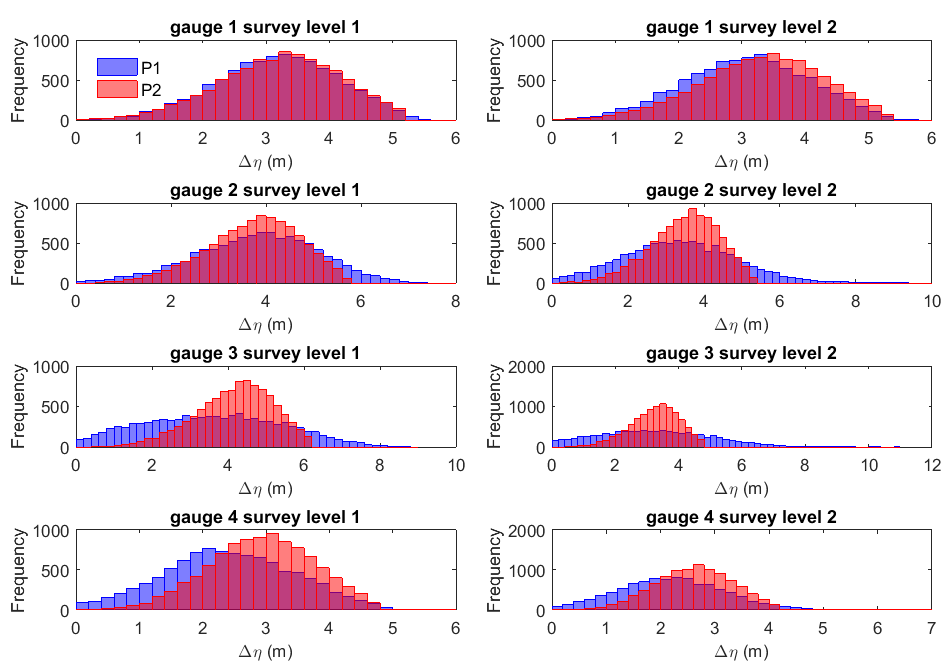}
\caption{Histogram of predictions of $10,000$ tsunami wave heights at four gauges due to the uncertain seabed uplift (prediction P2). Prediction P1 also accounts for the uncertainties in the bathymetry. Left column: high-resolution bathymetry survey (level 1); Right column: coarse bathymetry survey (level 2).}
\label{hist_pred}
\end{figure}

\section{Discussion}
\label{sec:dis}
We proposed a joint framework for emulation of high-dimensional simulators with dimension reduction. The gKDR approach is employed to construct low-dimensional approximations to the simulators. The approximations retain most of the information about the input-output behavior and make the emulation much more efficient. Both theoretical properties and numerical studies have demonstrated the efficiency and accuracy of the proposed approach and its advantages over other dimension reduction techniques. Our method can be applied for many purposes of uncertainty quantification such as risk assessment, sensitivity analysis and calibration, with great perspectives in real world applications.

There are some practical issues when applying the proposed framework. The hyperparameters in gKDR and the number of dimensions to be included in the emulator need to be specified properly. In practice, a simple trial and error procedure could be applied, especially when the results are not very sensitive to the choices. Cross validation steps could also benefit from parallel computing technique. The sample size also affect the predictive ability of the final emulator, as sufficient samples are needed to estimate the dimension reduction accurately. A diagnostic plot of predictive errors with increasing number of sizes such as \cref{rrsscv_res} could help identify the convergence. After determining the dimension reduction, a sufficient number of training samples with a proper design are often required to train the emulator in order to balance the computational cost and accuracy. The benefits of our approach are multiple. One can tackle uncertainty quantification tasks for complex models where boundary conditions are of high dimension. Beyond tsunami modeling, in climate simulation, weather forecasting or geophysical sciences, uncertainty quantification studies would become tractable and potentially offer solutions to important scientific problems.

\bibliographystyle{siam}
\bibliography{references_all}

\end{document}